\documentclass{article}

\usepackage{amssymb,amsmath,amsthm,sectsty,url,ifpdf}
\usepackage[usenames,dvipsnames]{xcolor}
\usepackage[colorlinks=true,linkcolor=NavyBlue,citecolor=teal]{hyperref}
\usepackage{boxedminipage}
\usepackage[boxed]{algorithm}
\usepackage{cleveref,aliascnt}
\usepackage{comment}

\usepackage[normalem]{ulem}
\usepackage[boxed]{algorithm}
\usepackage[labelfont=bf,small]{caption}
\floatname{algorithm}{Figure}
\usepackage{multicol}
\usepackage{graphicx}
\graphicspath{{content}}
\graphicspath{{content}}

\theoremstyle{plain}
\newtheorem{theorem}{Theorem}[section]
\newtheorem{lemma}[theorem]{Lemma}

\theoremstyle{definition}
\newtheorem{definition}[theorem]{Definition}

\newtheorem*{proofsketch}{Proof Sketch}

\begin{document}

\title{Computational Complexity of Game Boy Games}

\author{
 Hayder Tirmazi\thanks{City College of New York. Email:
    \texttt{hayder.research@gmail.com}}
\and
Ali Tirmazi
\and
Tien Phuoc Tran
}

\date{}

\maketitle

\begin{abstract}
    We analyze the computational complexity of several popular video games released for the Nintendo Game Boy video game console. We analyze the complexity of generalized versions of four popular Game Boy games: Donkey Kong, Wario Land, Harvest Moon GB, and Mole Mania. We provide original proofs showing that these games are \textbf{NP}-hard. Our proofs rely on Karp reductions from four of Karp's original 21 \textbf{NP}-complete problems: \textsc{Sat}, \textsc{3-Cnf-Sat}, \textsc{Hamiltonian Cycle}, and \textsc{Knapsack}. We also discuss proofs easily derived from known results demonstrating the \textbf{NP}-hardness of Lock `n' Chase and The Lion King.
\end{abstract}

\tableofcontents

\section{Introduction}\label{sec:introduction}

Formal analysis of the computational complexity of video games and puzzles has emerged as a vibrant research area~\cite{mit_hardness_book, GPC, aloupis_nintendo, viglietta}. This work analyzes multiple popular games released for the Nintendo Game Boy~\cite{gameboy}. Our results demonstrating the \textbf{NP}-hardness of Donkey Kong, Wario Land, and Mole Mania are new. The \textbf{NP}-Completeness of Simplified Harvest Moon has been discussed in prior work~\cite{Braeger2012SimplifiedHM}, but the publication's content was restricted to $140$ characters. We provide a more detailed proof demonstrating the \textbf{NP}-hardness of (not Simplified) Harvest Moon GB. We restrict our analysis to games released for the original Nintendo Game Boy\footnote{All products, company names, brand names, trademarks, sprites, and artwork are properties of their respective owners.
Sprites and artwork from games are used here under Fair Use for the educational purpose of illustrating mathematical theorems.}. We do not include games not supported by the original Game Boy and only supported by later devices such as the Game Boy Color and the Game Boy Advanced.

\subsection{Preliminaries}\label{sec:preliminaries}

We analyze generalized versions of the Game Boy games we mentioned. All the assumptions we make below are consistent with seminal work in this area by Aloupis et al.~\cite{aloupis_nintendo} and Viglietta~\cite{viglietta} (Section 1 in both papers). We only generalize the map size and leave all other game mechanics unchanged. We assume that the memory and screen size of the game are not bounded. We also assume that the game rules are as the game developers intended and ignore the existence of game-breaking glitches. Like Aloupis et al.~\cite{aloupis_nintendo}, we take the position that glitches are not an inherent part of the game; they are the result of an imperfect implementation of the game.

We introduce common definitions and notation that will be used throughout our work below. When we refer to a Game Boy game with title ``Example Game'' as \textit{Example Game}, we are referring to the name of the actual game. When we instead refer to it as \textsc{Example Game}, we are really referring to the decision problem underlying the game. We will define these decision problems for each game we analyze. We will also use the following common definitions for a Game Room and a Game Level.

\begin{definition}\label{def:game_room}
    Let a Game Room be a 2-dimensional map. Each cell on the map may contain zero or more game tiles. 
\end{definition}
\noindent By a ``cell'' we mean each entry or slot on the map. For example, a $2 \times 2$ map with tiles of size $1 \times 1$ consists of $4$ cells as it has $4$ slots where you can put tiles. Tiles may only be placed at discrete coordinates e.g. a tile may be placed at coordinate $(1, 0)$ but not at coordinate $(1.3, 0.9)$. Each game we discuss has a tile representing the avatar that can be controlled by the player using the game pad. We refer to this tile as the \textit{player} tile. Let $R$ be the set of all Game Rooms in a given Game Boy game. We model player movement between game rooms as follows.

\begin{definition}\label{def:game_transition}
    Let a Game Transition be a function $f: R \times R \mapsto R \times R$. $f$ takes two rooms, $s, d$ (source and destination) such that $s$ contains a player tile and $d$ does not. $f$ returns two rooms, $s^{\prime}, d^{\prime}$ that are identical two $s, d$ in every way other than the fact that $s^{\prime}$ does not contain a player tile and a player tile is placed in $d^{\prime}$.
\end{definition}

\begin{definition}\label{def:game_level}
    Let a Game Level be a 2-tuple $(R, T)$ where $R$ is a set of Game Rooms and $T$ is a set of Game Transitions.
\end{definition}

\subsection{Related Work}\label{sec:related_work}

There are several interesting papers on the computational complexity of video games. Demaine et al.~\cite{demaine_tetris} analyze the complexity of Tetris. Kaye~\cite{kaye_minesweeper} shows that Minesweeper is \textbf{NP}-complete. Cormode~\cite{cormode_lemmings} studies the complexity of Lemmings. Fori\v{s}ek~\cite{forisek} investigates the computational complexity of many 2-dimensional platformers including Commander Keen, Crystal Caves, Secret Agent, Bio Menace, Jill of the Jungle, Hocus Pocus, Duke Nukem, Crash Bandicoot, Jazz Jackrabbit, Lemmings, and Prince of Persia. Viglietta~\cite{viglietta} analyzes many video games released in the $1980$s and $1990$s including Pac-Man, Lode Runner, and Boulder Dash. We use one of Viglietta's meta-theorems in Section~\ref{sec:wario_land} to show that Wario Land is \textbf{NP}-hard via a Karp reduction from \textsc{Hamiltonian Cycle}. 

Aloupis et al.~\cite{aloupis_nintendo} prove that many of Nintendo's classic games are NP-hard. While Aloupis et al. analyze the complexity of games in the Donkey Kong franchise, they restrict their attention to Donkey Kong Country $1-3$. Aloupis et al.'s proof of \textbf{NP}-hardness relies on gadgets built using a game enemy called Zinger. This proof is not applicable to Donkey Kong released for the Game Boy as Zingers are not part of the enemy list of the game~\cite{donkeykongwiki}. We provide an original proof showing the \textbf{NP}-hardness of Donkey Kong for the Game Boy that relies two other game elements called switches and slide boards.

Braeger~\cite{Braeger2012SimplifiedHM} has a Tiny Transactions on Computer Science paper\footnote{Tiny Transactions on Computer Science is a venue which restricts the content of a research paper to $140$ characters} on the complexity of Harvest Moon. Braeger suggests that ``Simplified Harvest Moon'', a version of Harvest Moon restricted only to the production and sale of crops, is \textbf{NP}-Complete. The full body content of Braeger's paper is: \textit{There exists a trivial bijection from all instances of BKP\footnote{By BKP they are referring to the Bounded Knapsack Problem} to all instances of
Simplified Harvest Moon (SHM). Therefore, SHM is NP-complete.} In Section~\ref{sec:harvest_moon_gb} we provide an original proof inspired by Braeger's result showing that Harvest Moon GB is \textbf{NP}-hard via a Karp reduction from (unbounded) \textsc{Knapsack}. 

The computational complexity of games is a widely studied area, an exhaustive coverage of which is beyond the scope of this paper. We point the reader to the ``MIT Hardness Book~\cite {mit_hardness_book}'' or ``Games, Puzzles, and Computation~\cite{GPC}'' for a structured introduction to the field. There are also detailed surveys of the area including Eppstein~\cite{eppstein} and Kendall at al.~\cite{KendallSurvey}. Lastly, Gualà et al.~\cite{isnphard} have hosted ``Complexity of Games'', a compendium that indexes known hardness results for games and puzzles.

\section{Results}\label{sec:gameboy_games}

Our hardness proofs are based on reductions from four of Karp's original 21 \textbf{NP}-complete problems~\cite{Karp72}. Section~\ref{sec:donkey_kong} first defines the \textsc{Donkey Kong} decision problem and then provides a detailed proof of the \textbf{NP}-hardness of \textsc{Donkey Kong} via a Karp reduction from \textsc{3-Cnf-Sat}. Section~\ref{sec:wario_land} defined the \textsc{Wario Land} decision problem shows the \textbf{NP}-hardness of \textsc{Wario Land} via reduction from \textsc{Hamiltonian Cycle}. Section~\ref{sec:harvest_moon_gb} defines the \textsc{Harvest Moon GB} decision problem and demonstrates its \textbf{NP}-hardness via reduction from \textsc{Knapsack}. Finally, Section~\ref{sec:mole_mania} defines the \textsc{Mole Mania} decision problem and proves the \textbf{NP}-hardness of \textsc{Mole Mania} via a reduction from \textsc{Sat} by way of a reduction from \textsc{Push-1}, another well-studied~\cite{push1} \textbf{NP}-hard game.

\begin{figure}[]
    \centering
    \includegraphics[width=0.9\linewidth]{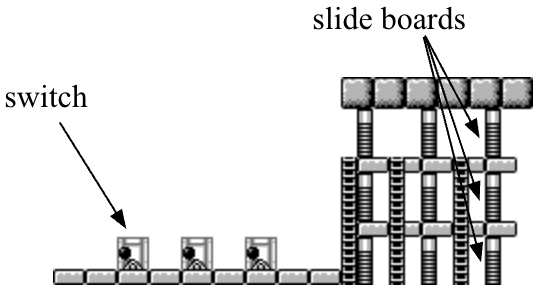}
    \caption{Any \textsc{3-CNF-Sat} instance can be converted to a level in Donkey Kong in polynomial time}
    \label{fig:donkey_kong_3_cnf_sat}
\end{figure}

\subsection{Donkey Kong}\label{sec:donkey_kong}

Donkey Kong was published for the Game Boy by Nintendo in $1994$. It is a puzzle platformer game. We note that the ``desert'' world in Donkey Kong contains switches the player can turn off and on to move one or more slide boards up or down~\cite{donkeykongmanual}. We first define the \textsc{Donkey Kong} decision problem and then show that this decision problem is \textbf{NP}-hard.

\begin{definition}\label{def:donkey_kong_decision_problem}
    An instance of the \textsc{Donkey Kong} decision problem is a Game Room (Definition~\ref{def:game_room}). For our results, we restrict our attention to the following tiles of the following types: 
    \[ 
    \{ \text{Empty}, \text{Switch}, \text{Slide-Board-Top}, \text{Slide-Board-Body}, \text{Brick-Floor}, \text{Block}, \text{Mario} \}\]
    
    \noindent We are also given two cells: $\{ \textsc{Start}, \textsc{Win} \}$. Mario is the player tile. For a given Game Room, if Mario, when placed at the \textsc{Start} cell, can reach the \textsc{Win} cell while obeying the rules of the game, then \textsc{Donkey Kong} outputs $1$. It outputs $0$ otherwise.
\end{definition}

\begin{theorem}\label{thm:donkey_kong}
\textsc{Donkey Kong} is \textbf{NP}-hard.
\end{theorem}

\begin{proof}

We demonstrate a Karp reduction from \textsc{3-Cnf-Sat} to \textsc{Donkey Kong}. For any \textsc{3-Cnf-Sat} instance with $n$ variables, construct a Game Room that begins with $n$ switches corresponding to each variable $x_{i}$ in the \textsc{3-Cnf-Sat} instance. We refer the reader to the left of Figure~\ref{fig:donkey_kong_3_cnf_sat}. For each $3$ literal disjunction $l_{1} \lor l_{2} \lor l_{3}$ in the \textsc{3-Cnf-Sat} formula, we construct vertically aligned slide boards accessible through a staircase. We refer the reader to the right of Figure~\ref{fig:donkey_kong_3_cnf_sat}. If $l_{1} = x_{i}$, the level opens the slide board at the top if and only if the $i^{\text{th}}$ switch from the left is turned on. Similarly, if $l_{1} = \lnot x_{i}$, the level opens the slide board if and only if the $i^{\text{th}}$ switch from the left is turned off. The middle and bottom slide boards are configured analogously, based on literals $l_{2}$ and $l_{3}$ respectively. Note that when we say a switch in the game is ``off'' we mean its handle is facing left, whereas a switch being ``on'' means its handle is facing right.

We select the left-most cell in the Game Room which is right above a tile of type Brick-Floor as the \textsc{Start} cell. Similarly, we pick the right-most cell in the Game Room that is above a tile of type Brick-Floor as the \textsc{Win} cell. With the switches and slide boards configured as described above, each valid solution of the \textsc{3-Cnf-Sat} instance corresponds to a valid way of scrolling through the Game Room from left to right. This reduction is valid because the Game Room can be constructed from a \textsc{3-Cnf-Sat} instance in polynomial time.
\end{proof}

\subsection{Wario Land}\label{sec:wario_land}

Wario Land: Super Mario Land 3 is a platform game published for the Game Boy by Nintendo in $1994$. Wario Land contains a door and key game mechanic that we use to derive our results. The game contains keys (also called skeleton keys) that are used to open a skull door (also called skeleton door). We refer the reader to Figure~\ref{fig:wario_land_keys}.

\begin{figure}[]
    \centering
    \includegraphics[width=\linewidth]{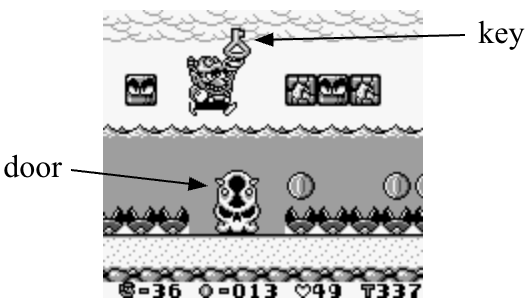}
    \caption{Wario Land contains a door and key game mechanic.}
    \label{fig:wario_land_keys}
\end{figure}

\begin{definition}\label{def:wario_skull_door}
    Let a \textit{skull door} be a game tile that can be in one of two states: $\{ \textsc{Open}, \textsc{Closed} \}$. If a skull door is in state \textsc{Open}, it enables a Game Transition (Definition~\ref{def:game_transition}) and transports Wario, the player tile, to another room in the game. 
\end{definition}

\begin{definition}\label{def:wario_key}
    Let a \textit{key} be a game tile that Wario can pick up in Wario Land. If Wario comes into contact with a skull door (Definition~\ref{def:wario_skull_door}) in state \textsc{Closed} and is carrying a key, Wario can change the state of the door to \textsc{Open}. If the key is used in this way, the key disappears as soon as the state of the skull door changes to \textsc{Open} i.e. the key cannot be reused.
\end{definition}

\begin{definition}
    Let a \textit{One-Way Game Transitions} be a pair of Game Rooms $(s, d)$ such that there exists a Game Transitions in the game that maps $(s, d)$ to $(s^{\prime}, d^{\prime})$ but there exists no Game Transition in the game that maps $(d, s)$ to $(d^{\prime}, s^{\prime})$. In other words, a one-way transition is a Game Transition that the player tile can traverse in one direction only. 
\end{definition}

\noindent To set the stage for our result, we show that One-Way Game Transitions can be constructed in Wario Land.

\begin{lemma}\label{lem:wario_one_way_path}
One-Way Game Transitions can be constructed in \textit{Wario Land}.
\end{lemma}

\begin{proof}
    In the following way, we can construct a One-Way Game Transition from any Game Room $r_{i}$ to any other Game Room $r_{j}$. Place a key and a skull door in $r_{i}$. The skull door transports Wario to an entry point in $r_{j}$. Place the entry point of the skull door in $r_{j}$ such that the accessible tiles below it are all further from the entry point than the maximum jump distance of Wario. We refer the reader to Figure~\ref{fig:wario_one_way_path}. When transported to Game Room $r_{j}$, Wario falls to the tiles below and cannot reach the entry point, making it impossible to return to Game Room $r_{i}$.
\end{proof}

\begin{figure}[]
    \centering
    \includegraphics[width=0.5\linewidth]{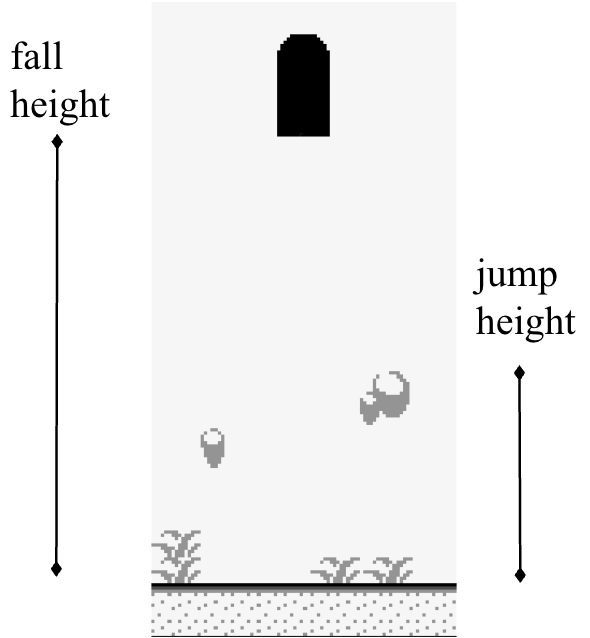}
    \caption{One-way paths can be constructed in Wario Land}
    \label{fig:wario_one_way_path}
\end{figure}

\begin{definition}
    Let a \textit{Skull Door Graph} be a directed planar graph the vertices of which have either 1 incoming edge and 2 outgoing edges, or 2 incoming edges and 1 outgoing edge.
\end{definition}

\noindent We prove the following trivial result, which will be useful in the reduction used in proving the \textbf{NP}-hardness of \textsc{Wario Land}. 

\begin{lemma}\label{lem:wario_vertex_exists}
    Let $G$ be any non-empty Skull Door Graph. $G$ contains at least one vertex with 2 incoming edges and 1 outgoing edge.
\end{lemma}

\begin{proof}
    Let $\alpha$ be the number of vertices with $1$ incoming and $2$ outgoing edges, and $\beta$ be the number of vertices with 2 incoming and 1 outgoing edge. The total incoming edges then are $\alpha + 2\beta$ and the total outgoing edges are $2\alpha + \beta$. By the handshaking lemma for directed graphs, we have $\alpha + 2\beta = 2\alpha + \beta$ which simplifies to $\alpha = \beta$. Since $G$ is non-empty, the result follows from the pigeonhole principle.
\end{proof}

\noindent We now define the \textsc{Wario Land} decision problem.

\begin{definition}\label{def:wario_land_decision_problem}
    An instance of the \textsc{Wario Land} decision problem is a Game Level in which each Game Room may contain a treasure that can be collected by Wario. One of the Game Rooms is marked as the starting Game Room. For a given Game Level, if Wario, when placed in the starting Game Room, can collect all treasures in the Game Level, while obeying the rules of the game, then \textsc{Wario Land} outputs $1$. It outputs $0$ otherwise.
\end{definition}

\begin{theorem}\label{thm:wario_land}
    \textsc{Wario Land} is \textbf{NP}-hard.
\end{theorem}

\begin{proofsketch}
We derive our theorem from a general strategy introduced by Viglietta~\cite{viglietta} (Metatheorem 3a in their work) that can be used to construct a Karp reduction from \textsc{Hamiltonian Cycle} to games with doors, keys, and one-way paths. 
\end{proofsketch}

\begin{proof}
    We construct a Karp reduction from \textsc{Hamiltonian Cycle} to \textsc{Wario Land}. Let $G$ be a (non-empty) Skull Door Graph with $n$ vertices. We note that \textsc{Hamiltonian Cycle} is \textbf{NP}-complete even for Skull Door Graphs~\cite{hamiltonian_cycle_planar_digraphs, viglietta}. We pick a vertex $v$ with $2$ incoming edges and $1$ outgoing edge. Lemma~\ref{lem:wario_vertex_exists} guarantees that at least one such $v$ exists. We attach a new outgoing edge (i.e. a second outgoing edge) to $v$ that connects it to a new vertex $u$. We then transform $G$ in the following way:
    \begin{itemize}
        \item We replace each vertex $i$ in $G$ with a Game Room $r_{i}$ containing a treasure.
        \item We replace each directed edge $(i, j)$ in $G$ except $(v, u)$ (picked above) with a One-Way Game Transition from Game Room $r_{i}$ to Game Room $r_{j}$ constructed using the method of Lemma~\ref{lem:wario_one_way_path}. 
        \item We replace directed edge $(v, u)$ with a One-Way Game Transition that does \textit{not} require a key i.e the skull door leading to the Game Room corresponding to $u$, that is placed in the Game Room corresponding to $v$, is already in state \textsc{Open}.
    \end{itemize}

    We initially put Wario in the Game Room corresponding to the vertex $v$ we originally picked. To win, Wario must collect the treasure in every Game Room. After our transformation, each Game Room except for the Game Room corresponding to $u$ contains exactly one key. Wario must use up a key picked up in every Game Room to access a door in state \textsc{Close} which leads to an unvisited Game Room. If Wario revisits any Game Room other than the Game Room corresponding to $v$, Wario can only enter doors in state \textsc{Open}. It follows that Wario can only win if Wario can visit every Game Room other than the one corresponding to $u$ and the first Game Room Wario revisits is the Game Room corresponding to $v$. From there, Wario can move to the Game Room that corresponds to $u$ to collect the final treasure. The Game Room corresponding to $u$ must necessarily be the final Game Room Wario visits as $u$ has no outgoing edges. Wario can win if and only if $G$ contains a Hamiltonian cycle. Therefore, since \textsc{Hamiltonian Cycle} is \textbf{NP}-complete, we have shown that \textsc{Wario Land} is \textbf{NP}-complete.
\end{proof}

\begin{figure}[]
    \centering
    \includegraphics[width=0.5\linewidth]{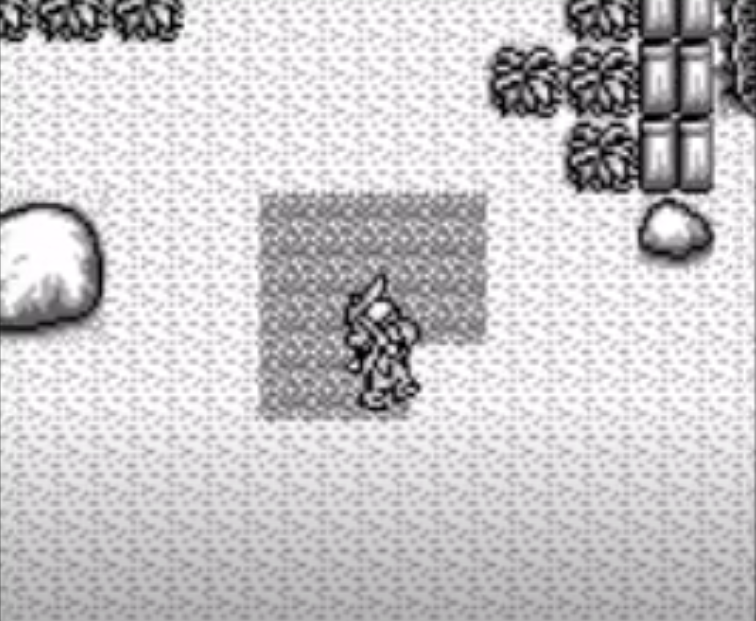}
    \caption{Crop planting and harvesting mechanics of Harvest Moon GB}
    \label{fig:harvest_moon_gb}
\end{figure}

\subsection{Harvest Moon GB}\label{sec:harvest_moon_gb}

Harvest Moon GB is a role-playing and farm simulation game published for the Game Boy in 1997. The game has a limited simulation time, measured in in-game days. The player is allocated a limited number of tiles in the game that can be used to grow crops.  Each crop takes a certain number of days to grow~\cite{HarvestMoonGBCrops}. We refer the reader to Figure~\ref{fig:harvest_moon_gb}. Each crop sells for a different price. Some crops are single-use i.e. they need to be re-planted from seeds after a harvest, while others can regrow without being re-planted. We now define the \textsc{Harvest Moon GB} decision problem.

\begin{definition}
    An instance of the \textsc{Harvest Moon GB} decision problem has $k$ tiles and a simulation time of $W$ days. It also has $n$ single-use crops each of which takes $w_{i}$ days to grow and sells for $v_{i}$. If the player can attain a revenue of at least $V$ in the provided $W$ days, \textsc{Harvest Moon GB} outputs $1$. Otherwise, it outputs $0$.
\end{definition}

\begin{theorem}\label{thm:harvest_moon_np_hard}
    \textsc{Harvest Moon GB} is \textbf{NP}-hard.
\end{theorem}

\begin{proof}
    We construct a Karp reduction from \textsc{Knapsack} to \textsc{Harvest Moon GB}. Consider an instance of \textsc{Knapsack} with maximum weight $W$ and inputs $\eta_{i}$ such that each $\eta_{i}$ has value $v_{i}$ and weight $w_{i}$. We construct a \textsc{Harvest Moon GB} instance where the player has exactly $1$ tile that can be used to plant crops and a maximum simulation time of $W$ days. For each $\eta_{i}$, we allow the player to grow a single-use crop that takes $w_{i}$ days to grow and sells for $v_{i}$. Thus, \textsc{Harvest Moon GB} outputs $1$ if and only if \textsc{Knapsack} outputs $1$. Therefore, since \textsc{Knapsack} is \textbf{NP}-complete, we have shown that \textsc{Harvest Moon GB} is \textbf{NP}-hard.
\end{proof}

\noindent The above result also holds for the original Harvest Moon game developed for the Super Nintendo Entertainment System (SNES) as it has identical farming simulation mechanics.

\subsection{Mole Mania}\label{sec:mole_mania}

Mole Mania is a puzzle game published by Nintendo for the Game Boy in $1996$. The player contains Muddy, who needs to navigate through levels. Each level has an above-ground component and an underground component. There are two types of floor tiles: $\{ \textsc{Soft}, \textsc{Hard} \}$. If Muddy is above a \textsc{Soft} floor tile, Muddy may dig it to reach the underground component of the level. Muddy may dig back up from the underground component to the above-ground component at another location. If Muddy is above a \textsc{Hard} floor tile, Muddy may \textbf{not} dig it, blocking the underground component. The game contains special tiles that Muddy can interact with. Balls, Barrels, and Cabbages are special tiles that Muddy can push and pull. Weights are special tiles that Muddy can only push, but not pull~\cite{molemaniamanual}. We refer the reader to Figure~\ref{fig:mole_mania}.

\begin{figure}[]
    \centering
    \includegraphics[width=0.45\linewidth]{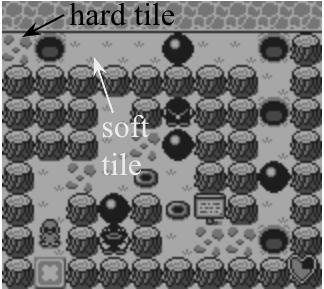}\hfill
    \includegraphics[width=0.5\linewidth]{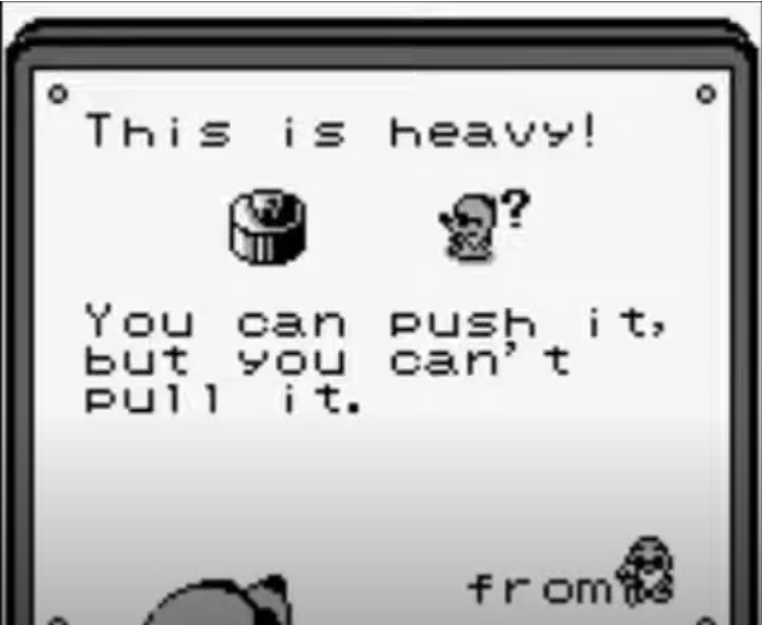}
    \caption{The game mechanics of Mole Mania. \textbf{Left}: the first-floor tile from the top-left (after the water) is a \textsc{Hard} tile. The third-floor tile from the top-left is a \textsc{Soft} tile. \textbf{Right}: Weight tiles are special tiles in Mole Mania that can only be pushed, but not pulled.}
    \label{fig:mole_mania}
\end{figure}

For our analysis of Mole Mania, we need a working understanding of Push-1, another puzzle game. We only consider the 2-dimensional version of Push-1. Push-1 consists of a planar square grid. Each cell of the grid either has a block tile or is empty. A robot is placed in one cell of the grid, and can move to any adjacent empty cell. The robot can push (but not pull) a block tile, and move it to the adjacent cell in the pushed direction if the adjacent cell is empty. One cell is marked as the winning cell. To win Push-1, the player must navigate the robot from its starting cell to the winning cell~\cite{push1website}. The \textsc{Mole Mania} and \textsc{Push-1} decision problems can be defined in the same way as the \textsc{Donkey Kong} decision problem (Definition~\ref{def:donkey_kong_decision_problem}).

\begin{lemma}\label{lem:push_1}
    \textsc{Push-1} is \textbf{NP}-hard via a Karp reduction from \textsc{Sat} to \textsc{Push-1}.
\end{lemma}
\begin{proof}
    This result was proved by Erik Demaine et al.~\cite{push1}. We refer the reader to Theorem 1 in their paper.
\end{proof}

\begin{theorem}\label{thm:mole_mania}
    \textsc{Mole Mania} is \textbf{NP}-hard.
\end{theorem}

\begin{proof}
    We construct a Karp reduction from \textsc{Sat} to \textsc{Mole Mania}. We first construct a Karp reduction from \textsc{Push-1} to \textsc{Mole Mania}. Consider an instance of \textsc{Push-1} with an $n \times n$ grid. We create a \textsc{Mole Mania} Game Room with $n \times n$ cells. We give each cell a \textsc{Hard} floor tile ensuring the game is restricted to the above-ground component. For each cell in the \textsc{Push-1} grid that contains a block, we place a corresponding cell in the \textsc{Mole Mania} Game Room at the same position with a Weight tile. We place Muddy Mole on the tile corresponding to the robot placement tile in \textsc{Push-1}. Similarly, we place the \textsc{Win} cell corresponding to the \textsc{Win} cell in \textsc{Push-1}. Now that we have a Karp reduction from \textsc{Push-1} to \textsc{Mole Mania}, we use Lemma~\ref{lem:push_1} and the well-known result that Karp reductions are transitive, to prove the existence of a Karp reduction from \textsc{Sat} to \textsc{Mole Mania}. 
\end{proof}

\section{Discussion}\label{sec:discussion}

We have covered non-trivial results for \textbf{NP}-hardness in Section~\ref{sec:gameboy_games}. Section~\ref{sec:known_results} discusses Game Boy games which are already known to be \textbf{NP}-hard. Section~\ref{sec:easily_derived_results} discusses Game Boy games that are not already known to be \textbf{NP}-hard, but for whom proving \textbf{NP}-hardness from already known results is relatively easy.

\subsection{Known \textbf{NP}-hard Games}\label{sec:known_results}
\textit{Tetris} was released for the Game Boy in $1989$. The gameplay is more or less identical to other implementations of \textit{Tetris} which are \textbf{NP}-hard~\cite{demaine_tetris, ThinTetris_JIP, TetrisPieces_FUN2024}. \textit{Mario's Picross} released for the Game Boy in $1995$ is a Nonogram puzzle. Ueda and Nagao~\cite{nonogram_npcomplete} prove that Nonogram puzzles are \textbf{NP}-complete. \textit{Pac-Man} was released for the Game Boy in $1990$. Viglietta~\cite{viglietta} shows that \textsc{Pac-Man} is \textbf{NP}-hard via a Karp reduction from \textsc{Hamiltonian Cycle}. Aloupis et al.~\cite{aloupis_nintendo} prove the \textbf{NP}-hardness of \textsc{Pokemon Red}, \textsc{Pokemon  Blue}, and \textsc{Pokemon Yellow} which were released for the Game Boy in the late $1990$s.

\subsection{Easily Derived Results}\label{sec:easily_derived_results}

\textit{Lock `n' Chase} released for the Game Boy in $1990$ has game mechanics similar to \textit{Pac-Ma}n.  We can construct a Karp reduction from an instance of \textsc{Pac-Man} to an instance of \textsc{Lock `n' Chase} and use the transitivity of Karp reductions to show that \textsc{Lock `n' Chase} is \textbf{NP}-hard. \textit{Lock 'n' Chase} also has key and door game mechanics which provides a second, even simpler, way to show that \textsc{Lock 'n' Chase} is \textbf{NP}-hard. We can construct a Karp reduction directly from \textsc{Hamiltonian Cycle} to \textsc{Lock `n' Chase} just like we did for Wario Land in Section~\ref{sec:wario_land}. 

\textit{The Lion King} was released for the Game Boy between $1994$ and $1995$. \textit{The Lion Game} features a bonus level featuring the character Timon where the player needs to collect ``bugs'' within a given time frame. Fori\v{s}ek~\cite{forisek} shows a metatheorem (Metatheorem 2 in their work) asserting that \textit{A 2D platform game where the collecting items feature is present and a time limit is present as a part of the instance is \textbf{NP}-hard}. The metatheorem relies on a Karp reduction from \textsc{Hamiltonian Cycle} on grid graphs. Fori\v{s}ek's metatheorem can easily be applied on the bonus level featuring Timon to demonstrate the \textbf{NP}-hardness of \textsc{The Lion King}.

\section{Open Problems}\label{sec:open_problems}

This work suggests the following two open problems.\\

\noindent
\textbf{Problem 1: Analyzing PSPACE-hardness.} We have only provided results involving the \textbf{NP}-hardness of the six Game Boy games we discussed. The question of whether or not any of these Game Boy games are also \textbf{PSPACE}-hard remains unsolved. Prior work~\cite{aloupis_nintendo, viglietta, forisek} has suggested frameworks for analyzing the \textbf{PSPACE}-hardness of video games which may prove helpful.\\

\noindent
\textbf{Problem 2: Complexity of Dr. Mario.} \textit{Dr. Mario} was released for the Game Boy in $1990$. While the game mechanics of \textit{Dr. Mario} are similar to those of \textit{Tetris}, it only has one kind of ``piece'' that falls: a two-colored medical capsule. We leave the \textbf{NP}-hardness of \textit{Dr. Mario} as an open problem. We note that the \textbf{NP}-hardness of \textit{Tetris} with the Super Rotation System (SRS) restricted to a single piece is a big open problem suggested by prior work~\cite{TetrisPieces_FUN2024}. We conjecture that these two problems are related.

\section{Conclusion}\label{sec:conclusion}

We have provided new results demonstrating that generalized versions of \textit{Donkey Kong}, \textit{Wario Land}, \textit{Mole Mania}, \textit{Lock `n' Chase}, and \textit{The Lion King}, are all \textbf{NP}-hard. We have also provided the first detailed proof of the \textbf{NP}-hardness of \textit{Harvest Moon GB}. Our constructions are consistent with the game mechanics of each game as long as the memory and screen size are unbounded.





\bibliographystyle{unsrt}
\bibliography{sn-bibliography}% common bib file

\begin{thebibliography}{10}

\bibitem{mit_hardness_book}
Erik~D. Demaine, William Gasarch, and Mohammad Hajiaghayi.
\newblock {\em Computational Intractability: A Guide to Algorithmic Lower Bounds}.
\newblock MIT Press, 2024.

\bibitem{GPC}
Robert~A. Hearn and Erik~D. Demaine.
\newblock {\em Games, Puzzles, and Computation}.
\newblock A K Peters, July 2009.

\bibitem{aloupis_nintendo}
Greg Aloupis, Erik~D. Demaine, Alan Guo, and Giovanni Viglietta.
\newblock Classic nintendo games are (computationally) hard.
\newblock {\em Theoretical Computer Science}, 586:135--160, 2015.
\newblock Fun with Algorithms.

\bibitem{viglietta}
Giovanni Viglietta.
\newblock Gaming is a hard job, but someone has to do it!
\newblock In Evangelos Kranakis, Danny Krizanc, and Flaminia Luccio, editors, {\em Fun with Algorithms}, pages 357--367, Berlin, Heidelberg, 2012. Springer Berlin Heidelberg.

\bibitem{gameboy}
Nintendo.
\newblock \textit{Nintendo {G}ame {B}oy}.
\newblock \url{https://www.nintendo.com/en-gb/Hardware/Nintendo-History/Game-Boy/Game-Boy-627031.html}.

\bibitem{Braeger2012SimplifiedHM}
Steven Braeger.
\newblock Simplified {H}arvest {M}oon is {NP}-{C}omplete.
\newblock {\em Tiny Trans. Comput. Sci.}, 1, 2012.

\bibitem{demaine_tetris}
Erik~D. Demaine, Susan Hohenberger, and David Liben-Nowell.
\newblock Tetris is {H}ard, {E}ven to {A}pproximate.
\newblock In Tandy Warnow and Binhai Zhu, editors, {\em Computing and Combinatorics}, pages 351--363, Berlin, Heidelberg, 2003. Springer Berlin Heidelberg.

\bibitem{kaye_minesweeper}
Richard Kaye.
\newblock Minesweeper is {NP}-complete.
\newblock In {\em The Mathematical Intelligencer}, 2000.

\bibitem{cormode_lemmings}
G.~Cormode.
\newblock The {H}ardness of the {L}emmings {G}ame, or {Oh} no, more {NP}-completeness proofs.
\newblock In {\em Proceedings of Third International Conference on Fun with Algorithms}, pages 65--76, 2004.

\bibitem{forisek}
Michal Fori\v{s}ek.
\newblock Computational complexity of two-dimensional platform games.
\newblock In {\em Proceedings of the 5th International Conference on Fun with Algorithms}, FUN'10, page 214–227, Berlin, Heidelberg, 2010. Springer-Verlag.

\bibitem{donkeykongwiki}
Mario Wiki.
\newblock Donkey {K}ong ({G}ame {B}oy).
\newblock \url{https://www.mariowiki.com/Donkey_Kong_(Game_Boy)}.

\bibitem{eppstein}
David Eppstein.
\newblock \textit{Computational {C}omplexity of {G}ames and {P}uzzles}.
\newblock \url{http://www.ics.uci.edu/~eppstein/cgt/hard.html}, 2009.

\bibitem{KendallSurvey}
Graham Kendall, Andrew~J. Parkes, and Kristian Spoerer.
\newblock A {S}urvey of {NP}-{C}omplete {P}uzzles.
\newblock {\em J. Int. Comput. Games Assoc.}, 31(1):13--34, 2008.

\bibitem{isnphard}
Luciano Gualà, Stefano Leucci, Matteo Almanza, Emilio Cruciani, Giorgio Ciotti, Luca~Di Donato, Arno Gobbin, Emanuele Natale, André Nusser, Roberto Tauraso, and Ben Wiederhake.
\newblock \textit{Complexity of {G}ames}.
\newblock \url{https://www.isnphard.com/}.

\bibitem{Karp72}
R.~Karp.
\newblock Reducibility among combinatorial problems.
\newblock In R.~Miller and J.~Thatcher, editors, {\em Complexity of Computer Computations}, pages 85--103. Plenum Press, 1972.

\bibitem{push1}
Erik~D. Demaine, Martin~L. Demaine, and Joseph O'Rourke.
\newblock {PushPush} and {Push-1} are {NP}-hard in {2D}.
\newblock In {\em Proceedings of the 12th Annual Canadian Conference on Computational Geometry (CCCG 2000)}, pages 211--219, Fredericton, New Brunswick, Canada, August 16--18 2000.

\bibitem{donkeykongmanual}
Nintendo.
\newblock \textit{Nintendo {D}onkey {K}ong ({G}ame {B}oy) {M}anual}.
\newblock \url{https://archive.org/details/donkey-kong-game-boy-manual/page/n5/mode/2up}.

\bibitem{hamiltonian_cycle_planar_digraphs}
J.~Plesn\'{\i}k.
\newblock The {NP}-completeness of the {H}amiltonian {C}ycle {P}roblem in {P}lanar {D}iagraphs with {D}egree {B}ound {T}wo.
\newblock {\em Inf. Process. Lett.}, 8(4):199–201, April 1979.

\bibitem{HarvestMoonGBCrops}
Harvest~Moon Wiki.
\newblock \textit{{C}rops}.
\newblock \url{https://harvestmoon.fandom.com/wiki/Crops_(GB)}.

\bibitem{molemaniamanual}
Nintendo.
\newblock \textit{Nintendo {M}ole {M}ania ({G}ame {B}oy) {M}anual}.
\newblock \url{https://archive.org/details/mole-mania-gb-manual-full-color/page/n7/mode/2up}.

\bibitem{push1website}
Erik Demaine.
\newblock Push-1.
\newblock \url{https://erikdemaine.org/push1/}.

\bibitem{ThinTetris_JIP}
Sualeh Asif, Michael Coulombe, Erik~D. Demaine, Martin~L. Demaine, Adam Hesterberg, Jayson Lynch, and Mihir Singhal.
\newblock Tetris is {NP}-hard even with {$O(1)$} rows or columns.
\newblock {\em Journal of Information Processing}, 28:942--958, 2020.

\bibitem{TetrisPieces_FUN2024}
{MIT Hardness Group}, Erik~D. Demaine, Holden Hall, and Jeffery Li.
\newblock Tetris with {F}ew {P}iece {T}ypes.
\newblock In Andrei~Z. Broder and Tami Tamir, editors, {\em Proceedings of the 12th International Conference on Fun with Algorithms (FUN 2024)}, volume 291 of {\em LIPIcs}, pages 24:1--24:18, La Maddalena, Italy, June 4--8 2024.

\bibitem{nonogram_npcomplete}
Nobuhisa Ueda and Tadaaki Nagao.
\newblock Np-completeness results for {NONOGRAM} via {P}arsimonious {R}eductions.
\newblock 1996.

\end{thebibliography}

% \section*{Appendices}

% \appendix

% \section{Additional Proofs}\label{sec:additional proofs}

\section{Acknowledgements}\label{sec:acknowledgements}

The first author would like to gratefully acknowledge Tien Phuoc Tran for his valuable feedback, editorial assistance, annotations, and intellectual contributions to this manuscript, and Ali Tirmazi, a talented young researcher currently in high school, whose innovative input was crucial to our understanding of the game mechanics we discussed as well as the plausibility of our constructions. The authors thank Erik Demaine for making his MIT course ``Algorithmic Lower Bounds: Fun with Hardness Proofs'' freely available online. The authors also thank Mohammad Hajiaghayi (and Erik Demaine again) who made their textbook~\cite{mit_hardness_book} freely available for us to learn from. The authors are grateful to Nintendo and the game developers of the Game Boy games we discussed. Finally, we thank various people in the video gaming community who put walkthroughs and long plays of the video games we analyzed on YouTube.

\end{document}